\documentclass{llncs}

\usepackage{latexsym}
\usepackage[dvipdfmx]{graphicx}
\usepackage{fancybox}
\usepackage{algorithm}
\usepackage[noend]{algpseudocode}

\usepackage{amsmath}
\usepackage{mediabb}
\usepackage{amssymb}

\def\unnumbered#1{%
\begin{trivlist}
\item[] {{#1}\kern0.5em}}
\def\endunnumbered{\end{trivlist}}

\def\myappenv#1#2{\begin{unnumbered}{{\bf #1~\ref{#2}.}}}
\def\endmyappenv{\end{unnumbered}}

\usepackage{url}
\usepackage{cite}

\def\newblock{\hskip .11em plus.33em minus.07em}

\newcommand{\name}[1]{\textit{#1}}
\newcommand{\set}[1]{\{#1\}}

\newcommand{\inset}[2]{\{#1\;|\;#2\}}

\newcommand{\Can}{{\it CAND}}

\newcommand{\adj}{{\rm adj}}

\newcommand{\MyProc}[1]{{\rm \textsc{#1}}}

\begin{document}

\title{Efficient Enumeration of Induced Subtrees in a K-Degenerate Graph}

\author{Kunihiro Wasa\inst{1} \and Hiroki Arimura\inst{1} \and Takeaki Uno\inst{2}}
\institute{Hokkaido University, Graduate School of Information Science and Technology, Japan, 
\email{\{wasa, arim\}@ist.hokudai.ac.jp}
\and National Institute of Informatics, Japan, 
\email{uno@nii.jp}}
\maketitle

\begin{abstract}
    In this paper, 
    we address the problem of  
    enumerating all induced subtrees 
    in an input \name{$k$-degenerate} graph, 
    where an \name{induced subtree} is 
    an acyclic and connected induced subgraph. 
    A graph $G = (V, E)$ is a $k$-degenerate graph 
    if for any its induced subgraph has a vertex 
    whose degree is less than or equal to $k$, 
    and many real-world graphs have small degeneracies, 
    or very close to small degeneracies. 
    Although, the studies are on subgraphs enumeration, 
    such as trees, paths, and matchings, 
    but the problem addresses the subgraph
    enumeration, such as enumeration of subgraphs that are trees.
    Their induced subgraph versions have not been studied well.
    One of few example is for chordless paths and cycles.
    Our motivation is to reduce 
    the time complexity close to $O(1)$ for each solution. 
    This type of optimal algorithms are proposed 
    many subgraph classes such as trees, and spanning trees. 
    Induced subtrees are fundamental object 
    thus it should be studied deeply 
    and there possibly exist some efficient algorithms. 
    Our algorithm utilizes nice properties of $k$-degeneracy 
    to state an effective amortized analysis. 
    As a result, the time complexity is reduced to 
    $O(k)$ time per induced subtree. 
    The problem is solved in constant time for each in
    planar graphs, as a corollary. 
\end{abstract}


\section{Introduction}

\name{Subgraph enumeration problems} are enumeration problems 
that given a graph $G$ and a graph class $\mathcal{S}$, 
output all subgraphs $S$ of $G$ satisfying $S \in \mathcal{S}$ without duplicates. 
Subgraph enumeration problems are widely studied~\cite{Avis:Fukuda:DAM:1996,Ferreira:Grossi:Rizzi:ESA:2011,Shioura:Tamura:Uno:SIAM:1997,Tarjan:Read:1975,Wasa:COCOON:2012,Birmele:etal:SODA:2013,Uno:SIGAL:2003,Tarjan:1973}. 
Enumeration involves a huge number of solutions, 
thus enumeration algorithms are supposed to run in short time, 
with respect to the number of solutions $N$. 
For example, if an algorithm runs in $O(Nf)$ time for small $f$, 
other than preprocessing, 
we can consider the algorithm is efficient. 
In this case, 
we say that the algorithm runs in $O(f)$ time per solution, 
or $O(f)$ time for each solution. 
Further, 
the maximum computation time between two consecutive outputs called \name{delay} is 
also considered as a more efficiency of enumeration algorithms. 
Note that delay will not be $O(f)$ 
even if an algorithm runs in $O(f)$ time per solution. 

Enumeration algorithms are widely studied in these days. 
Especially, 
the data mining area has a large amount of studies on pattern mining problem. 
The algorithms have to deal with huge databases and a huge number of solutions, 
thus there are great needs of the algorithm theory on efficient enumeration. 
As we show below, 
many recent studies focus on the development of small complexity algorithms. 
Compared to other algorithms, 
enumeration algorithms have some unique aspects. 
For example, 
by operating only on the differences between the solutions, 
one can develop algorithms that 
run in time shorter than the amount of exact output. 
Other than this, 
since the recursion is much more structured compared to optimization, 
we can develop a non-trivial amortized analysis. 
As a consequent, 
researches on the numeration algorithms have great interests. 

In what follows, 
we fix the input graph $G = (V, E)$, 
and let $m = |E|$, $n = |V|$. 
In the 1970s, 
Tarjan and Read~\cite{Tarjan:Read:1975} studied 
a problem of enumerating spanning trees in the input graph. 
Their algorithm runs in $O(m+n+mN)$ time. 
Shioura, Tamura, and Uno~\cite{Shioura:Tamura:Uno:SIAM:1997} 
is improved the complexity to $O(n+m+N)$ time. 
Tarjan~\cite{Tarjan:1973} proposed an algorithm for enumerating all cycle 
in $O((|V|+|E|)(|\mathcal{C}(G)|+1))$ time, 
where $\mathcal{C}(G)$ is all cycle in $G$. 
Birmel\'{e} \textit{et al.}~\cite{Birmele:etal:SODA:2013} 
improved the complexity to 
in $O(m + \sum_{c\in\mathcal{C}(G)}|c|)$ total time. 
They also presented an enumeration algorithm for
all st-paths in the input graph $G$ 
in $O(m + \sum_{\pi\in\mathcal{P}_{st}(G)}|\pi|)$ total time, 
where $\mathcal{P}_{st}(G)$ is all st-paths in $G$. 
Ferreira \textit{et al.}~\cite{Ferreira:Grossi:Rizzi:ESA:2011} proposed 
an enumeration algorithm 
that enumerating all subtree having exactly $k$ edges in $G$ in $O(kN)$ time. 
Wasa \textit{et al.}~\cite{Wasa:COCOON:2012} presented 
an improved version of Ferreira \textit{et al.}'s problem in constant time delay 
when the input is a tree. 
As we see, 
speed up of enumeration algorithms have been intensively studied in long history.

Compared to these studies, 
induced subgraph enumerations have not been studied well.
Avis and Fukuda~\cite{Avis:Fukuda:DAM:1996} considered 
the connected induced subgraph enumeration problem. 
Their algorithm is based on reverse search, 
and runs in $O(mnN)$ time.
Uno~\cite{Uno:SIGAL:2003} proposed 
an enumeration algorithm for enumerating all chordless path 
connecting the given vertices $s$ and $t$ 
and all chordless cycle in $O((m+n)N)$ time. 

In this paper, 
we address the problem of enumerating all induced subtrees
in the given graph, 
where an induced subtree is a connected induced subgraph that has no cycle. 
Assume that the set of vertices in an induced subtree is $S$. 
Then, $V \setminus S$ is a feedback vertex set of $G$. 
Feedback vertices are also 
fundamental graph objects and their enumeration problem is 
equivalent to that of induced subtrees. 
If the input graph $G$ is a tree, 
the connected induced subgraph of $G$ is a subtree. 
Thus, 
Wasa \textit{et al.}'s shows that 
the induced subtree enumeration problem can be solved in constant time delay 
when the input graph is a tree. 
Tree is a simple graph class, 
so we are motivated whether we can do better
in more general graph classes with non-trivial algorithms.

As a main result of this paper, 
we propose an algorithm for the $k$-degenerate graph case. 
The algorithm runs in $O(k)$ time per solution, 
after $(|V|+|E|)$ preprocessing time. 
The algorithm starts from the empty subgraph,
and adds a vertex recursively to enlarge the induced subtree.
The vertex to be added has to be adjacent to the current induced subtree, 
and has not to make a cycle. 
By using the degeneracy, 
we efficiently maintain the addible vertices, 
and the time complexity is bounded by a sophisticated amortized analysis.
Real world graphs usually have small degeneracies, 
or only few vertex removals result small degeneracies, 
the algorithm is expected to be efficient in practice. 
Compared to other graph classes, 
this is a strong point of $k$-degenerate graphs. 
There have been not so many studies on the use of the degeneracy for enumeration algorithm, 
and thus our approach introduces 
one of new way of developing practically efficient 
and theoretically supported algorithms.

The rest of this paper is organized as follows: 
In Section~\ref{sec:prelim}, 
we gives definitions in this paper 
and the definition of our problem. 
In Section~\ref{sec:algorithm}, 
we propose a basic enumeration algorithm based on a binary partition method. 
In Section~\ref{sec:time:complexity}, 
we improve the algorithm by using a property of the degeneracy, 
and analyze its time complexity. 
Finally, we conclude this paper and give future works in Section~\ref{sec:conc}.

\section{Preliminaries}
\label{sec:prelim}

\subsection{Graphs}

Let $G=(V, E)$ be an \name{undirected graph}, 
where $V$ is the set of \name{vertices} and $E \subseteq V^2$ is the set of \name{edges}. 
In this paper, we assume that $G$ is simple and finite. 
We denote by $(u, v)$ the edge connecting $u$ and $v$. 
For any vertices $u, v$ of $V$, 
we say that $u$ and $v$ are \name{adjacent} to each other if $(u, v) \in E$.  
We denote by $N_G(u)$ the set of all vertices adjacent to $u$ in $G$. 
We define the \name{degree} $d_G(u)$ of $u$ in $V$ as the number of vertices adjacent to $u$. 
In what follows,
if it is clear from context, 
we omit the subscript $G$. 

A \name{path} in $G$ is 
a sequence of distinct vertices $\pi(u, v) = (v_1 = u, \dots, v_j = v)$, 
such that $v_i$ and $v_{i+1}$ are adjacent to each other for $1 \le i < j$. 
If there is $\pi(u, v)$ in $G$, 
we say that the path \name{connects} $u$ and $v$. 
The \name{length} of path $\pi(u, v)$ is the number of vertices in $\pi(u, v)$ minus one. 
For any path $\pi(u, v)$ of length larger than one, 
$\pi(u, v)$ is called a \name{cycle} if $u = v$. 
We say that $G$ is \name{connected} 
if there is a path connecting any pair of vertices in $G$. 
$G$ is a \name{tree} if $G$ has no cycle and is connected.

\subsection{Induced subtrees}

Let $S$ be a subset of $V$. 
We denote by $G[S] = (S, E[S])$ the graph \name{induced} by $S$, 
where $E[S] = \inset{(u, v) \in E}{u, v \in S}$. 
We call $G[S]$ an \name{induced subgraph} of $G$. 
If no confusion, we regard $S$ as $G[S]$.
$|S|$ is the size of $S$. 
We say that $S$ is an \name{induced subtree} (see Fig.~\ref{img:induced:subtree}),  
if $S$ is a tree.  
In the following, we state the problem of this paper. 

\spnewtheorem*{wproblem}{Problem}{\itshape}{\itshape}
\begin{wproblem}[Induced subtree enumeration problem]
    \label{prob:main}
    Enumerate all induced subtrees in  $G = (V, E)$. 
\end{wproblem}

\begin{figure}[t]
    \begin{center}
        \includegraphics[width=10em]{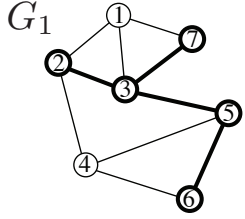}
        \caption{
            An induced subtree $S_1$ in $G_1$. 
            In the figure, bolded vertices and edges represent vertices and edges in $S_1$. 
            $S_1$ consists of $\set{2, 3, 5, 6, 7}$. 
            $S_1$ is an induced subtree in $G_1$ since $S_1$ is connected and acyclic. 
        }
        \label{img:induced:subtree}
    \end{center}
\end{figure}


\subsection{$K$-degenerate graphs}

A graph $G$ is \name{$k$-degenerate}~\cite{Lick:White:CJM:1970}
if any its induced subgraph of $G$ has 
a vertex whose degree is less than or equal to $k$. 
The \name{degeneracy} of $G$ is defined 
as the smallest $k$ satisfying the definition of $k$-degenerate graphs. 
Examples of graph classes with constant degeneracy include 
trees, grid graphs, outerplanar graphs, and planer graphs, 
thus degenerate graph is a large class of sparse graphs. 
These degeneracy are 1, 2, 2, and 5, respectively. 

From the definition of $k$-degeneracy, 
we obtain a vertex sequence $(u_1, \dots, u_{|V|})$ 
satisfying the condition    
\begin{equation*}
\forall 1\le i\le |V|,\; |\inset{u_j \in N(u_i)}{i < j \le |V|}| \le k \cdots (\star).  
\end{equation*}
This condition $(\star)$ implies that 
there exists an \name{ordering} among vertices of $G$ 
such that for any vertex $u$, 
the number of vertices adjacent to $u$ larger than it is at most $k$. 
Hereafter we assume that the vertices 
are indexed in this ordering. 
We say $u<v$ ($u>v$, respectively) 
if the index of $u$ is smaller than $v$ ($u$ is larger than $v$, respectively) with respect to this ordering. 
In Fig.~\ref{img:k:degeneracy}, 
we show an example of the ordering
satisfying $(\star)$. 
Matula and Beck~\cite{Matula:Beck:1983} proposed 
an algorithm for obtaining the degeneracy of $G$ and the ordering satisfying $(\star)$. 
By iteratively choosing the smallest degree vertex and removing it from $G$,  
their algorithm finds such an ordering in $O(|V| + |E|)$ time.

\section{Basic Binary Partition Algorithm}
\label{sec:algorithm}

\subsection{Candidate Sets and Forbidden Sets}
\label{subsec:can}

Let $S$ be an induced subtree of $G$. 
We define the \name{adjacency} of a vertex $u \in V$ to $S$ 
as $\adj(S, u) = |S\cap N(u)|$, 
that is, $\adj(S, u)$ is the number of vertices of $S$ adjacent to $u$. 

\begin{lemma}
    \label{lem:neighbor}
    Let $S$ be any induced subtree in $G$ and $u$ be any vertex $V \setminus S$. 
    $S \cup \set{u}$ is an induced subtree if and only if $\adj(S, u) = 1$. 
\end{lemma}

\begin{proof}
    If $\adj(S, u) > 1$, 
    $u$ is adjacent to two vertices $v$ and $w$ of $S$. 
    Since $S$ has a path $\pi$ connecting $v$ and $w$, 
    the addition of $u$ yields a cycle in $S\cup\set{u}$. 
    If $\adj(S, u) = 0$, $S\cup\set{u}$ is disconnected. 
    If $\adj(S, u) = 1$, $S\cup\set{u}$ is connected. 
    Since the degree of $u$ in $G[S\cup\set{u}]$ is one, 
    $u$ is not included in a cycle. 
    Thus, $G[S\cup\set{u}]$ does not contain a cycle. 
    \qed
\end{proof}

In each iteration, 
we maintain the \name{forbidden set} $X$ 
as the vertex set such that any vertex $u$ in $X$ 
satisfies either $u$ belongs to $S$, 
$S\cup\set{u}$ includes a cycle, 
or $u$ is forbidden to include in the solution by some ancestor iterations of the iteration.  
We also maintain the \name{candidate set} $\Can$ 
as the set of vertices whose additions yield induced subtrees and are not included in $X$. 
We maintain $\Can$ and $X$ for efficient computation. 
From Lemma~\ref{lem:neighbor}, 
they are disjoint, 
and for any vertex $u$, 
if $\adj(S, u) > 0$, 
$u$ belongs to either $\Can$ or $X$. 

\begin{figure}[!t]
    \begin{center}
        \includegraphics[width=28em]{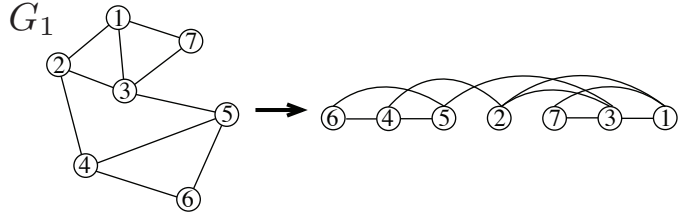}
        \caption{
            An example of an ordering of $G_1 = (V_1, E_1)$. 
            In the right graph, 
            vertices are sorted by the ordering that satisfies $(\star)$. 
        \label{img:k:degeneracy}}
    \end{center}
\end{figure}

\subsection{Basic Binary Partition}
\label{subsec:algorithm}

Our algorithm starts from the empty induced subtree $S = \emptyset$. 
In each iteration given an induced subtree $S$, 
we remove a vertex $u$ from $\Can$, 
and partition the problem into two; 
enumeration of all induced subtrees including $S\cup\set{u}$, 
and those including $S$ but not including $u$. 
We recursively do this partition 
until there is no vertex in $\Can$. 
The former can be solved 
by a recursive call with setting $S$ to $S\cup\set{u}$. 
The latter is solved 
by a recursive call with setting $X$ to $X\cup\set{u}$. 
In this way, 
we can enumerate all induced subtrees. 
We present 
the main routine \MyProc{ISE} of our algorithm 
in Algorithm~\ref{alg:ISE}.  
We show how to update candidate sets and forbidden sets in the next two lemmas. 

\begin{lemma}
    \label{lem:can:update}
    For an induced subtree $S$ and a vertex $u \in \Can$, 
    when we add $u$ to $S$ and remove $u$ from $\Can$, 
    $\Can$ changes to  
    \begin{equation*}
        (\Can\setminus N(u)) \cup (N(u) \setminus (\Can \cup X))).  
    \end{equation*}
\end{lemma}

\begin{proof}
    Any vertex in $\Can$ other than $N(u)$ 
    remains in $\Can$ after the addition of $u$ to $S$ 
    since the adjacencies of the vertices do not change. 
    If vertices in $N(u) \cap (\Can \cup X)$ are added to $S\cup\set{u}$, 
    they are in $S$, or they make cycles 
    since they are adjacent to $u$ and other vertices in $S$.  
    The adjacency of any vertex in $N(u) \setminus (\Can \cup X)$ is zero for $S$, 
    and one for $S\cup\set{u}$. 
    Any vertex $v \notin S$ satisfying $\adj(S\cup\set{u}, v) = 1$ 
    is either in $N(u)$ or $\Can$. 
    Thus, the statement holds. 
    \qed
\end{proof}

\begin{lemma}
    \label{lem:X:update}
    For an induced subtree $S$ and a vertex $u \in \Can$, 
    when we add $u$ to $S$ and remove $u$ from $\Can$, 
    $X$ changes to 
    \begin{equation*}
        X \cup \set{u} \cup  (\Can \cap N(u)).  
    \end{equation*}
\end{lemma}

\begin{proof}
    Any vertex $v \in X$ remains in $X$ for $S\cup\set{u}$, 
    since $\adj(S\cup\set{u}, v) \ge \adj(S, v)$ always holds. 
    From the definition of the forbidden set, 
    $u$ is in $X$ for $S\cup\set{u}$. 
    Further, 
    any vertex $v$ in $\Can \cap N(u)$ makes cycles when they are added to $S\cup\set{u}$, 
    since $\adj(S\cup\set{u}, v) \ge 2$ holds. 
    By adding $u$ to $S$, 
    no other vertex is forbidden to be added, 
    thus the statement holds. 
    \qed
\end{proof}

\begin{theorem}
    \label{thm:correctness}
    Algorithm \MyProc{ISE} enumerates all induced subtrees 
    in the input graph $G = (V, E)$ without duplicates. 
\end{theorem}

\begin{algorithm}[t]
    \caption{Main routine \MyProc{ISE}: Enumerating all induced subtrees in $G$}
    \label{alg:ISE}
    \begin{algorithmic}[1]
        \Procedure{ISE}{$G = (V, E), S, \Can, X$}
        \State \textbf{if} $\Can = \emptyset$ 
               \textbf{then} output $S$; 
               \textbf{return};  
        \label{alg:ISE:output}
        \State choose the smallest vertex $u$ from $\Can$ and remove $u$ from $\Can$; 
        \label{alg:ISE:choose}
        \State call $\MyProc{ISE}(G, S, \Can, X\cup\set{u})$; 
        \label{alg:ISE:call1}
        \State call $\MyProc{ISE}(G, S\cup\set{u}, (\Can \setminus N(u)) \cup (N(u) \setminus \Can), X\cup\set{u}\cup(\Can \cap N(u)))$; 
        \label{alg:ISE:call2}
        \EndProcedure
    \end{algorithmic}
\end{algorithm}

\section{Improved Binary Partition Algorithm}
\label{sec:time:complexity} 

From Lemma~\ref{lem:can:update} and Lemma~\ref{lem:X:update}, 
we can easily see that 
the computation time of updating 
the candidate set and the forbidden set 
is $O(d_G(u))$ by checking all vertices adjacent to $u$. 
However, in this way, 
we must check some vertices again and again. 
Specifically, 
let us assume $u, v$ are consecutively added to $S$, 
and $w \notin S$ is adjacent to $u$, $v$ and another vertex in $S$. 
When we add $u$ to $S$, 
we check whether we can add $w$ to the candidate set of $S\cup\set{u}$. 
After generating $S\cup\set{u}$, we check $w$ again 
when we add $v$ to $S\cup\set{u}$. 
In order to avoid this redundant checking, 
we improve the way of updating the candidate set and the forbidden set 
by using the following set.

\begin{definition}
    Suppose that $u$ is a vertex of $\Can$ for an induced subtree of $G$. 
    We define a set $\Gamma(u, X)$ 
    as follows: 
    \begin{equation*}
        \Gamma(u, X) = \inset{v\in N(u)}{v \notin X, v < u}. 
    \end{equation*}
\end{definition}

\begin{lemma}
    \label{lem:mod:update:Can}
    Let $S$ be an induced subtree of $G$, 
    $u$ be the smallest in the candidate set $\Can$ of $S$, 
    and $X$ be the forbidden set of $S$. 
    Then, the following formula holds: 
    \begin{equation*}
         N(u) \setminus (\Can \cup X) 
     = (N'(u) \setminus (\Can \cup X)) \cup \Gamma(u, X), 
    \end{equation*}
    where $N'(u) = \inset{v \in N(u)}{u < v}$.  
\end{lemma}

\begin{proof}
    Let $Z$ be the set of vertices larger than $u$. 
    Since $u$ is the smallest vertex in $\Can$, 
    $(N(u)\setminus(\Can\cup X))\cap Z = (N'(u)\setminus (\Can\cup X))$. 
    From the definition of $\Gamma(u, X)$ 
    and $u$ is the smallest in 
    $\Can$, 
    $(N(u)\setminus(\Can\cup X))\cap(V\setminus Z) = N''(u)\setminus (\Can \cup X) = (N''(u) \setminus \Can) \cap (N''(u) \setminus X) = \Gamma(u, X)$, 
    where $N''(u) = \inset{v\in N(u)}{v < u}$. 
    This concludes the lemma.
    \qed
\end{proof}

In what follows, 
we use an 
adjacency lists for the sets $\Can$, $X$, and $\Gamma$, 
so that a removal and the recover of the removed element 
can be done in $O(1)$ time, 
and the merge of two sets can be done in linear time of their sizes.

\begin{lemma}
    \label{lem:mod:update:X:neighbor}
    When we add a vertex $u$ to $X$, 
    the update of  $\Gamma(v, X)$ for all vertices $v$ is done in $O(k)$ time.
\end{lemma}

\begin{proof}
    To update, 
    it is suffice to remove $u$ from $\Gamma(v, X)$ from all $v>u$. 
    Thus, it takes $O(k)$ time.
    \qed
\end{proof}

\begin{lemma}
    \label{lem:mod:update:Can:X}
    Let $S$ be an induced subtree of $G$, 
    $u$ be the smallest in the candidate set $\Can$ of $S$, 
    and $X$ be the forbidden set of $S$.  
    When we add $u$ to $S$ and remove $u$ from $\Can$, 
    the computation time of updating $\Can$ and $X$ are
    $O(k + |\Gamma(u, X)|)$ and $O(k)$ time, respectively.  
\end{lemma}

\begin{proof}
    Since $u$ is the smallest vertex in $\Can$, 
    $|\Delta| \le k$, 
    where $\Delta = |\Can \cap N(u)|$. 
    Since vertices in $N(u)$ are sorted by the ordering, 
    the computation time of $\Delta$ is $O(k)$. 
    Thus, 
    adding vertices in $\Delta$ and $u$ to $X$ 
    and removing $\Delta$ from $\Can$ are done in $O(k)$ time. 
    From Lemma~\ref{lem:mod:update:Can}, 
    since $|\inset{v \in N(u)}{u < v}| \le k$, 
    the computation time of adding these vertex to $\Can$ is $O(k + |\Gamma(u, X)|)$. 
    Hence, the lemma holds. 
    \qed 
\end{proof}

\begin{figure}[t]
    \begin{center}
        \includegraphics[width=35em]{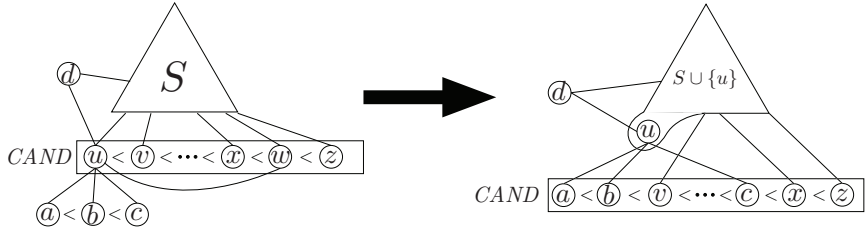}
        \caption{
            This figure shows the changes between candidate set $\Can$ 
            by the addition of $u$ to $S$. 
            $S$ is an induced subtree and 
            $\set{u, v, \dots, x, w, z}$ is the candidate set of $S$. 
            Let assume that $a < b < u < c$ and $d < u$. 
            Since $d$ does not belongs to $\Gamma(u, X)$, 
            $d$ is skipped checking. 
        }
        \label{img:candidate}
    \end{center}
\end{figure}

In Fig.~\ref{img:candidate}, 
we show the changes of between 
the candidate set of $S$ and that of $S\cup\set{u}$ 
after adding $u$ to $S$. 
We implement $\Can$ and $X$  by doubly linked lists. 
Thanks to the doubly linked list, 
the cost for a deletion and a recover of a vertex 
can be done in constant time.

\begin{theorem}
    \label{thm:complexity}
    Let $G = (V, E)$ be the input graph and $k$ is the degeneracy of $G$. 
    Our algorithm enumerates all induced subtrees in $G$ 
    in $O(k)$ time per solution 
    after $O(|V|+|E|)$ preprocessing time without duplicates 
    using $O(|V|+|E|)$ space. 
\end{theorem}

\begin{proof}
    Since the update of $\Can$ and $X$ is correct, 
    the correctness of the algorithm is obvious. 
    (I) We discuss the time complexity of the preprocessing. 
    First, our algorithm computes an ordering of vertices 
    by Matula and Beck's algorithm~\cite{Matula:Beck:1983} 
    in $O(|V| + |E|)$ time. 
    Next, our algorithm sorts 
    vertices belonging to each adjacency list by using a bucket sort.
    Thus, the preprocessing time is $O(|V|+|E|)$. 

    (II) 
    We consider an iteration inputting $S$, $X$, and $\Can$, and assume
    that $\Can'$ is the candidate set for $S\cup\set{u}$. 
    Line~\ref{alg:ISE:output} and line~\ref{alg:ISE:choose} run in $O(1)$ time. 
    From Lemma~\ref{lem:mod:update:X:neighbor}, 
    line~\ref{alg:ISE:call1} needs $O(k)$ time. 
    From Lemma~\ref{lem:mod:update:Can:X}, 
    since it is clear that $|\Gamma(u, X)| < |\Can'|$,   
    our algorithm needs $O(k + |\Can'|)$ time for computing $\Can'$ and $X$. 
    The update of $\Gamma$'s is done in $O(k|\Can\cap N(u)|)$ time, from Lemma~\ref{lem:mod:update:X:neighbor}. 
    We observe that for each vertex $w$ 
    such that $v\in (N(u)\cap \Can)$ is removed from $\Gamma(w, X)$, 
    $w$ is in $\Can$ of $S\cup\set{v}$, 
    that will be generated by a descendant of this iteration.
    We charge the cost of constant time to remove $v$ from $\Gamma(w, X)$ 
    to the induced subtree $S\cup\set{v, w}$. 
    Then, we can see that $S\cup\set{v,w}$ is charged only from iterations inputting $S$, 
    that divides the problem by $u'$ such that $(u',v)\in E$, 
    that is, the iteration generates $S\cup \set{u'}$.
    We consider the average amount of the charge over all induced subtrees of $S\cup\set{v,w}$, 
    $v\in \Can$, and $w$ is in $\Can$ of $S\cup\set{v}$. 
    Since the number of pairs $\set{u,v}\subseteq \Can$ is at most $k|\Can|$, 
    we can see the average charge is $O(k)$ for each $S\cup\set{v,w}$. 
    Thus, in summary, we can see the update time for $\Gamma$ in an iteration is bounded by $O(k)$, on average.
    Thus, an iteration takes $O(k + k |\Can'|)$ time on average.
    We observe that the sum of $|\Can'|$ over all iterations is no greater 
    than the sum of $|\Can|$ over all induced subtrees, 
    since $\Can'$ is the candidate set of $S\cup\set{u}$ and forbidden set $X\cup\set{u}$, 
    and $S\cup\set{u}$ is generated only from $S$. 
    Further, we can see that $S\cup\set{u}$ is generated only from $S$ 
    this iteration.
    Hence, 
    thus the sum of $|\Can|$ over all induced subtrees is bounded 
    by the number of induced subtrees. 
    Therefore, 
    the computation time for each iteration is bounded by $O(k)$ on average. 

    In a binary partition algorithm, 
    each iteration at the leaf of the recursion 
    outputs a solution, 
    and each non-leaf iteration generates exactly two recursive calls. 
    Thus, the number of iterations (recursive calls) of a binary partition algorithm 
    is at most $2N$. 
    Hence, the computation time per induced subtree is $O(k)$. 
    All all sets the algorithm maintains are of size $O(|V|+|E|)$ in total. 

    We need a bit care to perform a recursive call. 
    When a recursive call is made, 
    we record the operations 
    to prepare the parameters given to the recursive call on the memory. 
    When the recursive call ends, 
    we apply the inverse operations of the recorded operations 
    to recover the variables such as $\Can$ and $X$. 
    In this way, 
    we can recover the variables from the updated ones 
    without increasing the time complexity. 
    Since no vertex is added or deleted from the same variable twice, 
    the accumulated space for the recorded operations is bounded by $O(|V|+|E|)$. 
    From the above arguments, 
    our algorithm runs in $O(k)$ time per solution 
    after $O(|V|+|E|)$ preprocessing time using $O(|V|+|E|)$ space. 
    \qed
\end{proof}


\section{Conclusion}
\label{sec:conc}

In this paper, 
we have presented an algorithm for enumerating all induced subtrees in $k$-degenerate graph. 
Our algorithm runs in $O(k)$ time per solution after linear preprocessing time 
using linear space. 
From this result, 
we obtain the following corollary; 
if the input graph has a constant degeneracy, 
our algorithm is optimal with respect to the computation time per solution. 
$K$-degenerate graphs often appear in real-world data 
even when with much noise. 
Thus considering the applications, 
it is important to study on efficient computation on $k$-degeneracy. 
This result is one of the first steps for such studies, 
and researches on enumeration algorithms on $k$-degenerate graphs will be an important issue.

\subsection*{Acknowledgement}
This work was partially supported by 
MEXT Grant-in-Aid for Scientific Research (A) $24240021$ 
and Grant-in-Aid for JSPS Fellows $25\cdot1149$.

\end{document}